\documentclass[11pt, letter]{article}
\usepackage{subcaption}
\usepackage[UKenglish]{babel}
\usepackage[utf8]{inputenc}
\usepackage[T1]{fontenc}
\usepackage{stmaryrd}
\usepackage{algorithmic}
\usepackage{marginnote}
\usepackage[ruled,vlined,linesnumbered]{algorithm2e}
\usepackage{booktabs,tabularx,threeparttable}
\usepackage[normalem]{ulem}

\usepackage{amsthm,amssymb,amsmath,amscd}

\usepackage{fullpage}
\usepackage{tikz}
\usetikzlibrary{fit,backgrounds,positioning}

\usepackage{latexsym}
\usepackage{graphicx}	
\usepackage[all]{xy}	
\usepackage{psfrag}	

\usepackage{todonotes}
\setuptodonotes{inline}

\newtheorem{lemma}{Lemma}
\newtheorem{theorem}{Theorem}

\newenvironment{restatetheorem}[1]
{\begingroup
\def\thetheorem{\ref{#1}}
\begin{theorem}}
{\end{theorem}
\addtocounter{theorem}{-1}
\endgroup}

\title{Cities at Play:\\ Improving Equilibria in Urban Neighbourhood Games}

\author
{
Martin Gairing\thanks{University of Liverpool: {\tt gairing@liverpool.ac.uk}}
\and Adrian Vetta\thanks{McGill University: {\tt adrian.vetta@mcgill.ca}}
\and Zhanzhan Zhao\thanks{Chinese University of Hong Kong, Shenzhen: {\tt zhanzhanzhao@cuhk.edu.cn}}
}

\begin{document}

\maketitle

\renewcommand{\thefootnote}{\arabic{footnote}}

\begin{abstract}
How should cities invest to improve social welfare when individuals respond strategically to local conditions? We model this question using a game-theoretic version of Schelling’s bounded neighbourhood model, where agents choose neighbourhoods based on concave, non-monotonic utility functions reflecting local population. While naïve improvements may worsen outcomes — analogous to Braess’ paradox — we show that carefully designed, small-scale investments can reliably align individual incentives with societal goals. Specifically, modifying utilities at a total cost of at most 
$0.81 \epsilon^2 \cdot \texttt{opt}$
guarantees that every resulting Nash equilibrium achieves a social welfare of at least 
$\epsilon \cdot \texttt{opt}$, where $\texttt{opt}$ is the optimum social welfare. Our results formalise how targeted interventions can transform supra-negative outcomes into supra-positive returns, offering new insights into strategic urban planning and decentralised collective behaviour.
\end{abstract}

\bibliographystyle{plain}

\section{Introduction}\label{sec:intro}

How, when, and where should society invest its resources to most effectively improve the utility of its citizens?
We investigate this question in the context of urban planning, focusing on neighbourhood improvement. 
For example, suppose a metropolis has a development budget of \$100 million to invest, say, in roadworks, playgrounds, parks, etc. Classical economics predicts the potential social benefit of this investment to range from
zero benefit (i.e., \$0 with the money wasted and better returned to taxpayers) to full benefit (i.e., \$100m and, possibly, slightly greater given a moderate multiplier effect).
Surprisingly, however, this massively understates the range of
outcomes that can arise from this investment. The reason being that it ignores the distinction between the macro-motives of the urban planner and the micro-motives of individual citizens. In particular, 
personal incentives can cause individuals to react to investments 
in ways that have positive or negative effects on the utilities of other citizens. In our setting, the consequent micro-behaviour 
of interest is the act of moving neighbourhoods,
and we quantify the effects of this using a standard variant of Schelling's model of segregation \cite{schelling1971dynamic}.

Schelling \cite{schelling1971dynamic} explores a dynamic model with two types of agents occupying nodes in a graph and having preferences on the composition of agent types in their neighbourhood. His work demonstrates how seemingly minor individual preferences for neighbours of their own kind can lead to extreme segregation, even when individuals would prefer some integration. Schelling's model applies to general graphs, but also explores the bounded-neighbourhood model where the city is divided into blocks and agents decide to enter or exit a block based on the overall demographic composition in the blocks. The bounded-neighbourhood model has been studied more widely (see e.g. \cite{gauvin2009phase,grauwin2009competition}), in particular Grauwin et al.~\cite{grauwin2009competition} introduced a simplification of Schelling's model where there is only one type of agent and agent utilities depend on the density of agents in their block. The authors show how results for their single population model can be readily extended to the two populations model of Schelling. The main contribution is
the analytical prediction of stationary states, which builds on the free-energy function used in physics and omits previously required numerical simulations. 

In this paper, we look at the same model as Grauwin et al.~\cite{grauwin2009competition} but from a game-theoretic perspective. Our study shows how to direct macroscopic investments in order to optimize social welfare in equilibrium, i.e., we align micro-motives with the social welfare objective.
One striking illustration of the negative effects of micro-motives
concerns Braess' Paradox~\cite{Bra68} where the construction of a new road in a city can make everyone's travel times increase!
This phenomenon arises because the new road can induce bottlenecks in the road network. Braess' paradox has been widely observed in practice.
Indeed, for a theoretic perspective it is potentially ``as common as not''~\cite{SZ83}.
Conversely, it is commonplace for what can be viewed as {\em disinvestments} (e.g. the closing of a road, traffic calming measures, road tolls, etc.) to increase traffic speeds for everyone.

On the other hand, micro-behaviours can produce disproportionate positive effects.
A celebrated example concerns the {\em Bilbao effect}~\cite{Plaza06, Plaza07}. 
The economic benefit to the local economy of the Guggenheim Museum in Bilbao is estimated to be \$650m annually. As a result, the micro-behaviours related to tourism and economic regeneration greatly exceeded the \$100m construction cost of the museum.

Thus, the set of outcomes achievable after a neighbourhood investment range from supra-negative to supra-positive.
We will present examples of both these situations in Section~\ref{sec:model} after we have formally defined our urban planning model.

The purpose of this paper is to ascertain how and when supra-positive investment returns are achievable. Our conclusion is quite startling: {\em either the current neighbourhood plan already proffers social welfare comparable to the optimal social welfare achievable at that time or there exist
targeted neighbourhood investments that will produce supra-positive returns}.

In order to formalise and quantify this conclusion, 
we utilise a generalised version of Schelling's bounded neighbourhood model~\cite{schelling1971dynamic}. 
We present the model in Section~\ref{sec:model} and
explain with examples the game theoretic tools
we will use to analyse it.
Then, in Section~\ref{sec:background}, we discuss related literature.
In Section~\ref{sec:upper}, we prove our main result, Theorem~\ref{thm:main}, that supra-positive investment returns always exist if the current neighbour equilibrium has low welfare in comparison to the optimal social welfare achievable. In Section~\ref{sec:lower}, we prove a complementary result, Theorem~\ref{thm:minor}, which states that the investment return guarantees given in Theorem~\ref{thm:main} are essentially the best possible, in that they are tight within a small multiplicative factor. 
We conclude in Section~\ref{sec:conc}.

\section{A Game Theoretic Model of Urban Planning}\label{sec:model}

For our urban planning model, we use a generalised version of the {\em bounded neighbourhood model} of Schelling~\cite{schelling1971dynamic}. There are $n$ neighbourhoods in a metropolis, 
where each neighbourhood $i$ has a maximum population capacity of $u_i$ and a population of $x_i\in [0, u_i]$. Thus, the total population of the metropolis is $N=\sum_{i=1}^n x_i$.
For the sake of a cleaner presentation, we use a \emph{non-atomic} model for the population, where the population consists of a continuum of infinitely many individuals, each having only an infinitesimal contribution to the population. This is a good approximation for our setting as populations are usually large and each individual has only a tiny impact on the utilities. Moreover, we stress that all our results and proofs also go through (mutatis mutandis) for the atomic version with finitely many individuals. 

Two remarks are in order here. First, the original model of Schelling~\cite{schelling1971dynamic} had a population with two races 
of population densities $d_i$ and $1-d_i$ in neighbourhood $i$, respectively.
For the purposes of analysis, this is equivalent to our setting; to see this,
let the population $x_i= d_i\cdot u_i$ correspond to the population of the first race and let the spare capacity $u_i-x_i= (1-d_i)\cdot u_i$ correspond to the population of the second race.  Second, our motivation applies at the scale of neighbourhoods in a metropolis. But our results apply at an arbitrary scale, ranging from blocks in a small town to states in a large country.

Now recall our interest lies in micro-motives, specifically the motivations
and actions of individual citizens. To study this, we assume that there is a possibly distinct utility for living in each neighbourhood $i$, which naturally depends on the total population (equivalently, population density) in that neighbourhood. Formally,
we have utility functions $f_1,f_2,\dots, f_n$ each on $[0,u_i]$. Here $f_i(x_i)$ denotes the utility of each individual living in neighbourhood $i$ when the population at that location is $x_i\in [0,u_i]$. Each function $f_i$ is assumed to be non-negative and concave, but importantly need not be monotonic. 

We remark it is very important to allow for non-monotonic utility functions.
In particular, a standard assumption is that the utility function is increasing at low population densities but decreasing at high densities (see e.g. \cite{BT08,PV07,Zhang04}).
This is due to negative externalities that arise both from under-population 
and from overcrowding.
For example, as the total population/population density gets too large then the negative consequences of overcrowding may include
excessive traffic, pollution, noise, crime, etc. as well as increased living costs resulting from excess demand in the neighbourhood.
At the other extreme, a population that is too small may be
insufficient to support necessary public services such as  transportation, schools, health services, etc. nor 
private enterprises such as restaurants, shops, cafes, sports facilities, etc.


\subsection{Nash Equilibria}\label{sec:NE}
Of course each individual can choose which neighbourhood to live in.
This induces a game with a dynamic in which individuals may sequentially 
move neighbourhoods. In such a game an equilibrium state is a Nash equilibrium, where no individual can improve its utility by moving neighbourhoods. 

Let's illustrate this with an example. In Figure~\ref{fig:NE} there
are $n=2$ neighbourhoods; we call the first neighbourhood ``red'' and the second neighbourhood ``blue''. In this example, we assume that the neighbourhood capacities are equal. Moreover, we set $u_1=u_2=1$. Furthermore, the utility functions, $f_1$ and $f_2$,
are identical with a utility peak when the population in the neighbourhood is~$\frac12$. 
Finally, we assume that the total population of the metropolis is $N=1$.
\begin{figure}[h!]
\centering
\newcommand{\dd}{0.04}
\begin{tikzpicture}[scale=0.77,xscale=5,yscale=4,domain=0:1]
  \draw[->] (-0.02,0) -- (1.05,0) node[below right] {$x_1$};
  \draw[->] (0,-0.02) -- (0,1.08+2*\dd) node[above] {$f(x_1)$};

  \draw (0,0.02) -- (0,-0.02) node[below] {$0$};
  \draw (0.5,0.02) -- (0.5,-0.02) node[below] {$\frac{1}{2}$};
  \draw (1,0.02) -- (1,-0.02) node[below] {$1$};
  
  \draw[color=red, ultra thick]   (0,\dd) -- (0.5,1);
  \draw[color=red, ultra thick]   (0.5,1) -- (1,0);
  \draw (0.02,\dd) -- (-0.02, \dd) node[left]{$\delta$};
  \draw (0.02,1) -- (-0.02, 1) node[left]{1};
  
  \draw (0.02,\dd) -- (-0.02, \dd) node[left]{$\delta$};

  \draw[thin, dashed, color=gray] (0.5,-0.02) -- (0.5,1.02+2*\dd);
  \draw[thin, dashed, color=gray] (-0.02,1) -- (1.02,1); 
\end{tikzpicture}
\quad
\begin{tikzpicture}[scale=0.77,xscale=5,yscale=4,domain=0:1]
  \draw[->] (-0.02,0) -- (1.05,0) node[below right] {$x_2$};
  \draw[->] (0,-0.02) -- (0,1.08+2*\dd) node[above] {$f(x_2)$};

  \draw (0,0.02) -- (0,-0.02) node[below] {$0$};
  \draw (0.5,0.02) -- (0.5,-0.02) node[below] {$\frac{1}{2}$};
  \draw (1,0.02) -- (1,-0.02) node[below] {$1$};
  
  \draw[color=blue, ultra thick]   (0,\dd) -- (0.5,1);
  \draw[color=blue, ultra thick]   (0.5,1) -- (1,0);
  \draw (0.02,\dd) -- (-0.02, \dd) node[left]{$\delta$};
  \draw (0.02,1) -- (-0.02, 1) node[left]{1};

  \draw (0.02,\dd) -- (-0.02, \dd) node[left]{$\delta$};

  \draw[thin, dashed, color=gray] (0.5,-0.02) -- (0.5,1.02+2*\dd);
  \draw[thin, dashed, color=gray] (1,-0.02) -- (1,1.02+2*\dd);
  \draw[thin, dashed, color=gray] (-0.02,1) -- (1.02,1); 
\end{tikzpicture}
\caption{An instance with an equilibrium of high social welfare.}\label{fig:NE}
\end{figure}
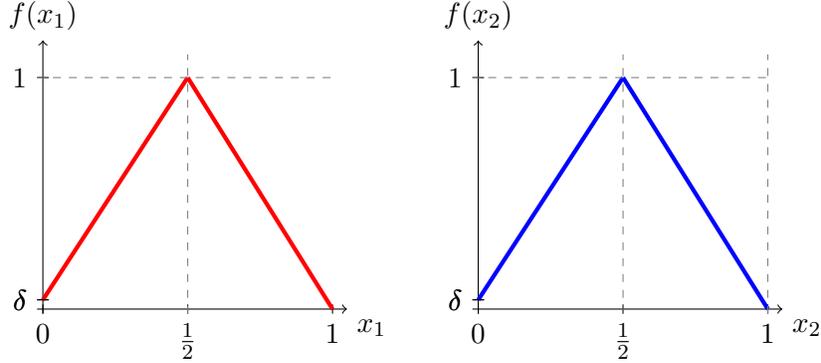

In this example, the Nash equilibrium is then $(x_1,x_2)=(\frac12,\frac12)$.
This is a feasible solution because $N=1=\sum_{i=1}^n x_i = \frac12 + \frac12$. Moreover, this is a Nash equilibrium because every agent receives a utility $1$ which is the maximum possible.
For example, an agent in the red neighbourhood cannot increase its utility by moving to the blue neighbourhood, and vice versa.
In fact, this Nash equilibrium gives the optimal achievable social welfare for this specific example.

\subsection{Social Welfare and the Price of Anarchy}\label{sec:PoA}
Unfortunately this is atypical. Equilibrium/stable states
in these urban planning games need not proffer high welfare for the 
agents. To illustrate this, we need the following definitions.
The {\em social welfare} when the populations are ${\bf x}=\{x_1, x_2,\dots, x_n\}$ is ${\tt wel}({\bf x})=\sum_{i=1}^n x_i\cdot f_i(x_i)$. 
The optimum achievable social welfare is denoted
${\tt opt} = \max\limits_{\bf x\in \mathcal{F}} \, {\tt wel}({\bf x})$, where
$\mathcal{F}$ denotes the set of all feasible neighbourhood population arrangements for the given urban planning instance.

Figure~\ref{fig:PoA} gives an example where the social welfare of the equilibrium state is very low in comparison to the optimal social welfare.
Notice that blue neighbourhood is now very slightly more attractive than the
red neighbourhood. Indeed the utility function $f_2$ now goes through the three points $(0, 3\delta)$, $(\frac12, 1 +2\delta)$ and $(1, 2\delta)$ rather than $(0, \delta)$, $(\frac12, 1)$ and $(1,0)$. For small~$\delta$ this example is nearly identical to that in Figure~\ref{fig:NE}, and the optimal solution remains
${\bf x}^*= (x^*_1,x^*_2)=(\frac12,\frac12)$ with social welfare ${\tt opt} = {\tt wel}({\bf x}^*)=\sum_{i=1}^n x^*_i\cdot f_i(x^*_i)
= \frac12\cdot 1 + \frac12\cdot (1+2\delta) = 1+\delta$.
However, the Nash equilibrium outcome is now dramatically different. 
To understand this, observe that at $(x^*_1,x^*_2)=(\frac12,\frac12)$
an agent in the red neighbourhood receives a utility of $1$ but can
increase its utility to (just under) $1+2\delta$ by moving to the blue neighbourhood. Moreover this behaviour is self-sustaining and the dynamics
cause every agent to move from the red neighbourhood to the blue neighbourhood.
Specifically, assume the current solution is $(x_1, x_2)=(\xi, 1-\xi)$.
Then an agent in the red neighbourhood has a utility of $f_1(\xi) = \delta+ 2\xi(1-\delta)$. Moving to the blue would give it a utility of 
$f_2(1-\xi) = 2\delta+ 2\xi$. But this is larger than $f_1(\xi)$, so the agent will make this move. Notice that the resultant dynamic can only terminate
at the unique Nash equilibrium of ${\bf x}=(x_1, x_2)=(0, 1)$.
But this Nash equilibrium, where everyone lives in the blue neighbourhood, has a social welfare of ${\tt wel}({\bf x})=\sum_{i=1}^n x_i\cdot f_i(x_i) = 0\cdot \delta + 1\cdot 2\delta = 2\delta$. Note that this is a terrible outcome (albeit the unique equilibrium) when $\delta$ is small.

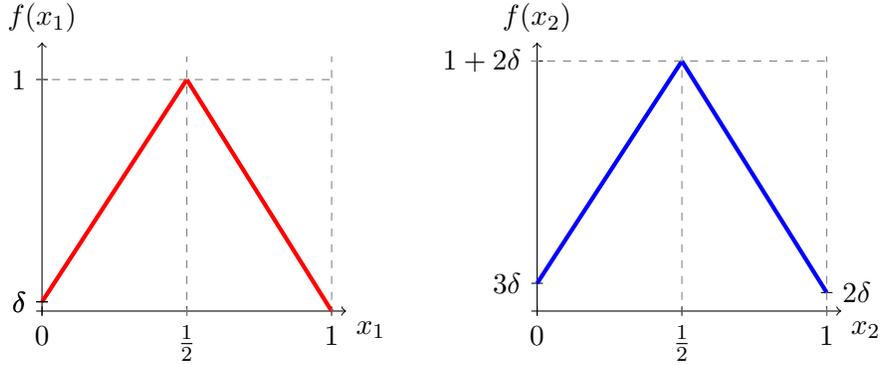
\begin{figure}[h!]
\centering
\newcommand{\dd}{0.04}
\begin{tikzpicture}[scale=0.77,xscale=5,yscale=4,domain=0:1]
  \draw[->] (-0.02,0) -- (1.05,0) node[below right] {$x_1$};
  \draw[->] (0,-0.02) -- (0,1.08+2*\dd) node[above] {$f(x_1)$};

  \draw (0,0.02) -- (0,-0.02) node[below] {$0$};
  \draw (0.5,0.02) -- (0.5,-0.02) node[below] {$\frac{1}{2}$};
  \draw (1,0.02) -- (1,-0.02) node[below] {$1$};
  
  \draw[color=red, ultra thick]   (0,\dd) -- (0.5,1);
  \draw[color=red, ultra thick]   (0.5,1) -- (1,0);
  \draw (0.02,\dd) -- (-0.02, \dd) node[left]{$\delta$};
  \draw (0.02,1) -- (-0.02, 1) node[left]{1};

  \draw (0.02,\dd) -- (-0.02, \dd) node[left]{$\delta$};

  \draw[thin, dashed, color=gray] (0.5,-0.02) -- (0.5,1.02+2*\dd);
  \draw[thin, dashed, color=gray] (1,-0.02) -- (1,1.02+2*\dd);
  \draw[thin, dashed, color=gray] (-0.02,1) -- (1.02,1); 
\end{tikzpicture}
\quad
\begin{tikzpicture}[scale=0.77,xscale=5,yscale=4,domain=0:1]
  \draw[->] (-0.02,0) -- (1.05,0) node[below right] {$x_2$};
  \draw[->] (0,-0.02) -- (0,1.08+2*\dd) node[above] {$f(x_2)$};

  \draw (0,0.02) -- (0,-0.02) node[below] {$0$};
  \draw (0.5,0.02) -- (0.5,-0.02) node[below] {$\frac{1}{2}$};
  \draw (1,0.02) -- (1,-0.02) node[below] {$1$};
  
  \draw (0.98,2*\dd) -- (1.02, 2*\dd) node[right]{$2\delta$};
  \draw (0.02,3*\dd) -- (-0.02, 3*\dd) node[left]{$3\delta$};
  \draw (0.02,1+2*\dd) -- (-0.02, 1+2*\dd) node[left]{$1+2\delta$};

  \draw[color=blue, ultra thick]   (0,\dd+2*\dd) -- (0.5,1+2*\dd);
  \draw[color=blue, ultra thick]   (0.5,1+2*\dd) -- (1,2*\dd);

  \draw[thin, dashed, color=gray] (0.5,-0.02) -- (0.5,1.02+2*\dd);
 \draw[thin, dashed, color=gray] (1,-0.02) -- (1,1.02+2*\dd);
  \draw[thin, dashed, color=gray] (-0.02,1+2*\dd) -- (1.02,1+2*\dd); 
\end{tikzpicture}
\caption{An instance with an equilibrium of low social welfare.}\label{fig:PoA}
\end{figure}

This example nicely illustrates the concept of the {\em price of anarchy},
the worst case ratio between the optimal social welfare and the
social welfare of a Nash equilibrium.\footnote{The price of anarchy is a very well-studied concept in game theory and computer science it measures the practical performance of a system against the optimal performance. The nomenclature derives from the idea that in settings where the price of anarchy is large, equilibria solutions where every agent acts individualistically (anarchy) may provide low social welfare. In such systems, it is critical to investigate whether or not improved outcomes can be obtained via coordination mechanisms or via slight, but practical, modifications to the underlying game.}
As shown, the price of anarchy of
this instance is $\frac{1+\delta}{2\delta} \approx \frac{1}{2\delta}$ which
tends to $\infty$ as $\delta$ tends to $0$.
Indeed the social welfare of the Nash equilibrium
satisfies ${\tt wel}({\bf x}) \le  2\delta\cdot {\tt opt}$.
That is, for small $\delta$, the Nash equilibrium delivers negligible social welfare in comparison to the optimal solution.

\subsection{A Braess-Type Paradox in Urban Planning}\label{sec:braess}

This potential for equilibria to result in catastrophically poor societal welfare outcomes is a major motivation for this work.
Specifically, this possibility has serious implications for our specific topic
of study, namely, targeted investments in urban planning.
To see this, consider the example in Figure~\ref{fig:supra-negative}.
Assume, to begin, that the red and blue neighbourhoods have identical
utility functions (the solid red and blue lines). Now imagine that after 
a small targeted investment the blue neighbourhood has improved very slightly,
resulting in the dashed blue utility function. Before the improvement
the situation is the same as in Figure~\ref{fig:NE} and, consequently, the social welfare at equilibrium is $1$.
But after the investment the situation is the same as in Figure~\ref{fig:PoA} and, hence, the social welfare at equilibrium has fallen dramatically to $2\delta$.
Ergo, we obtain an extreme form of Braess Paradox, the micro-behaviours that result from an improvement in the neighbourhood can lead to a calamitous decline in social welfare.

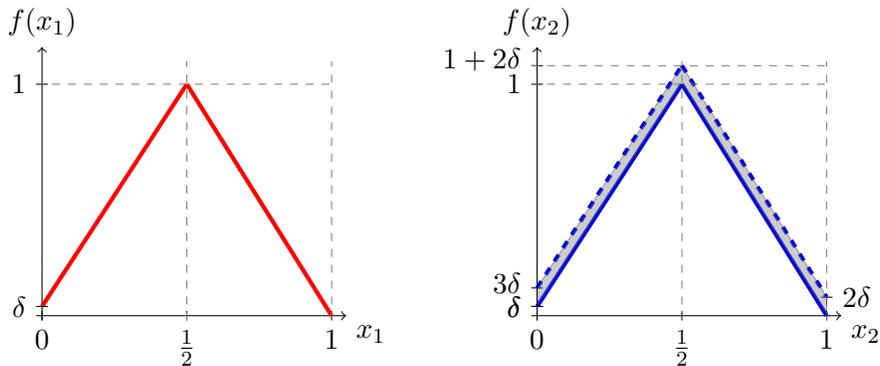
\begin{figure}[h!]
\centering
\newcommand{\dd}{0.04}
\begin{tikzpicture}[scale=0.77,xscale=5,yscale=4,domain=0:1]
  \draw[->] (-0.02,0) -- (1.05,0) node[below right] {$x_1$};
  \draw[->] (0,-0.02) -- (0,1.08+2*\dd) node[above] {$f(x_1)$};

  \draw (0,0.02) -- (0,-0.02) node[below] {$0$};
  \draw (0.5,0.02) -- (0.5,-0.02) node[below] {$\frac{1}{2}$};
  \draw (1,0.02) -- (1,-0.02) node[below] {$1$};
  
  \draw[color=red, ultra thick]   (0,\dd) -- (0.5,1);
  \draw[color=red, ultra thick]   (0.5,1) -- (1,0);
  \draw (0.02,\dd) -- (-0.02, \dd) node[left]{$\delta$};
  \draw (0.02,1) -- (-0.02, 1) node[left]{1};

  \draw[thin, dashed, color=gray] (0.5,-0.02) -- (0.5,1.02+2*\dd);
  \draw[thin, dashed, color=gray] (1,-0.02) -- (1,1.02+2*\dd);
  \draw[thin, dashed, color=gray] (-0.02,1) -- (1.02,1);
\end{tikzpicture}
\quad
\begin{tikzpicture}[scale=0.77,xscale=5,yscale=4,domain=0:1]
  \draw[->] (-0.02,0) -- (1.05,0) node[below right] {$x_2$};
  \draw[->] (0,-0.02) -- (0,1.08+2*\dd) node[above] {$f(x_2)$};

  \draw (0,0.02) -- (0,-0.02) node[below] {$0$};
  \draw (0.5,0.02) -- (0.5,-0.02) node[below] {$\frac{1}{2}$};
  \draw (1,0.02) -- (1,-0.02) node[below] {$1$};
  
  \draw[color=blue, ultra thick]   (0,\dd) -- (0.5,1);
  \draw[color=blue, ultra thick]   (0.5,1) -- (1,0);
  \draw (0.02,\dd) -- (-0.02, \dd) node[left]{$\delta$};
  \draw (0.98,2*\dd) -- (1.02, 2*\dd) node[right]{$2\delta$};
  \draw (0.02,3*\dd) -- (-0.02, 3*\dd) node[above left = -2mm and 0mm]{$3\delta$};
  \draw (0.02,1) -- (-0.02, 1) node[left]{1};
  \draw (0.02,1+2*\dd) -- (-0.02, 1+2*\dd) node[above left = -2mm and 0mm]{$1+2\delta$};

  \draw[color=blue,ultra thick,dashed]   (0,\dd+2*\dd) -- (0.5,1+2*\dd);
  \draw[color=blue,ultra thick,dashed]   (0.5,1+2*\dd) -- (1,2*\dd);
  \draw (0.02,\dd) -- (-0.02, \dd) node[left]{$\delta$};

  \draw[thin, dashed, color=gray] (0.5,-0.02) -- (0.5,1.02+2*\dd);
  \draw[thin, dashed, color=gray] (1,-0.02) -- (1,1.02+2*\dd);
  \draw[thin, dashed, color=gray] (-0.02,1) -- (1.02,1);
  \draw[thin, dashed, color=gray] (-0.02,1+2*\dd) -- (1.02,1+2*\dd); 

   \draw [fill=gray, opacity=0.4]
       (0, 3*\dd) -- (0.5,1+2*\dd)  -- (1,2*\dd) -- (1,0) -- (0.5, 1) -- (0, \dd) -- cycle;
\end{tikzpicture}
\caption{An instance with supra-negative investment returns.}\label{fig:supra-negative}
\end{figure}

Accordingly, it is vital to understand the dynamic effects caused by micro-behaviours in urban planning. 
Most notably, it induces the following question: {\em is it possible to make small targeted investments
to ensure that every resultant equilibrium has high social welfare?}
This is the key question we examine in this paper and,
remarkably, we answer it in the affirmative.

\subsection{Our Results}\label{sec:results}

Our goal is to make targeted investments in a selection of neighbourhoods to ensure that the social welfare in the metropolis at the resulting equilibrium is comparable to the optimal solution.
Moreover, we want the total cost of these targeted investments to be small.
To quantify the cost, we assume that a targeted investment in neighbourhood $i$ changes the utility function from $f_i$ to $g_i$. 
Moreover, the cost of the investment is proportional to the area of the symmetric difference $|f_i \oplus g_i|$, i.e., the area between the two functions. 
It follows that the total cost of all the targeted improvements made is $\sum_{i=1}^n |f_i \oplus g_i|$.

For example, in Figure~\ref{fig:supra-negative} the cost of the targeted investment is proportional to the shaded area in the blue neighbourhood.
Observe that this area is $2\delta$ which is negligible in comparison to the optimal social welfare of ${\tt opt}=1+\delta$.
Our main result is that a targeted investment of that small
a magnitude is always sufficient to ensure an equilibrium outcome that is comparable to the optimal
solution. (Observe that in Figure~\ref{fig:supra-negative}, zero investment
was needed to ensure this as the original equilibrium was the optimal solution!)

Before stating the main theorem, we must explain why the
assumption that investment cost is proportional to the 
area of the symmetric difference is natural for urban planning settings.
First, evidently, the larger the increase in the utility function (from $f_i$ to $g_i$), the higher the cost. 
Thus, the improvement cost is a function of the improvement of the utility function (i.e. the area of the symmetric difference).
But why would this cost
also be independent of where along the $x$-axis the improvements occur?
The reason is that improvement costs are roughly independent of congestion
levels in the neighbourhood. Specifically, the cost of building 
schools, roads, playgrounds, traffic calming measures, etc.
may be slightly more in a densely populated neighbourhood than in
a sparsely populated neighbourhood, but never more than a constant 
multiplicative factor greater. Thus we may assume that investment costs are proportional to the area of the symmetric difference.\footnote{Moreover, we can relax this assumption and allow investment costs to be dependent on population density, within a constant multiplicative factor. This results presented in this paper then still hold subject to that multiplicative factor.}

Our main theorem, which we prove in Section~\ref{sec:upper}, is that to ensure a social welfare of at least $\epsilon\cdot {\tt opt}$ only requires an investment of at most $0.81 \epsilon^2\cdot {\tt opt}$.

\begin{theorem}\label{thm:main}
In any urban planning instance, there exist targeted investments with a total cost at most $0.81 \epsilon^2\cdot {\tt opt}$ that ensure every resulting equilibrium has welfare at least $\epsilon\cdot {\tt opt}$.
\end{theorem}

To illustrate, assume that $\epsilon=\frac{1}{10}$ and our targeted investments cost \$100 million. The resulting social welfare at the new equilibrium will be at least \$1.235 billion.
Thus, if the initial equilibrium had low welfare then we obtain supra-positive returns from our targeted investments. Indeed, observe that obtaining
the same welfare via cash handouts to individuals rather than targeted investments would cost the taxpayer over a billion dollars!

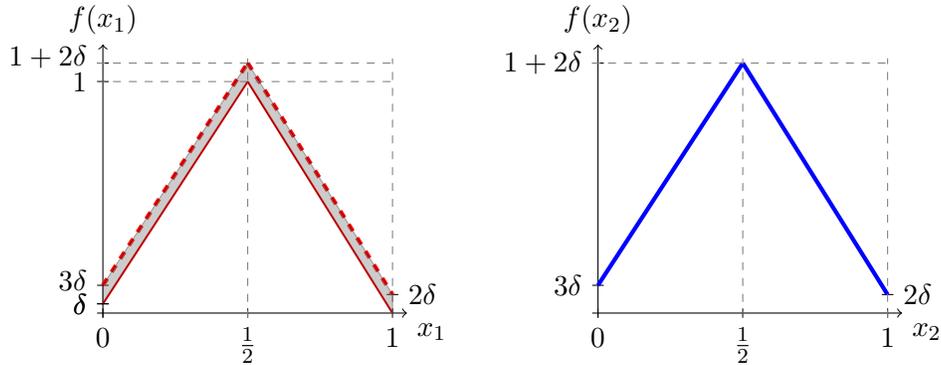
\begin{figure}[h!]
\centering
\newcommand{\dd}{0.04}
\begin{tikzpicture}[scale=0.77,xscale=5,yscale=4,domain=0:1]
  \draw[->] (-0.02,0) -- (1.05,0) node[below right] {$x_1$};
  \draw[->] (0,-0.02) -- (0,1.08+2*\dd) node[above] {$f(x_1)$};

  \draw (0,0.02) -- (0,-0.02) node[below] {$0$};
  \draw (0.5,0.02) -- (0.5,-0.02) node[below] {$\frac{1}{2}$};
  \draw (1,0.02) -- (1,-0.02) node[below] {$1$};
  
  \draw[color=red, thick]   (0,\dd) -- (0.5,1);
  \draw[color=red, thick]   (0.5,1) -- (1,0);
  \draw (0.02,\dd) -- (-0.02, \dd) node[left]{$\delta$};
  \draw (0.98,2*\dd) -- (1.02, 2*\dd) node[right]{$2\delta$};
  \draw (0.02,3*\dd) -- (-0.02, 3*\dd) node[above left = -2mm and 0mm]{$3\delta$};
  \draw (0.02,1) -- (-0.02, 1) node[left]{1};
  \draw (0.02,1+2*\dd) -- (-0.02, 1+2*\dd) node[above left = -2mm and 0mm]{$1+2\delta$};

  \draw[color=red,ultra thick,dashed]   (0,\dd+2*\dd) -- (0.5,1+2*\dd);
  \draw[color=red,ultra thick,dashed]   (0.5,1+2*\dd) -- (1,2*\dd);
  \draw (0.02,\dd) -- (-0.02, \dd) node[left]{$\delta$};

  \draw[thin, dashed, color=gray] (0.5,-0.02) -- (0.5,1.02+2*\dd);
  \draw[thin, dashed, color=gray] (1,-0.02) -- (1,1.02+2*\dd);
  \draw[thin, dashed, color=gray] (-0.02,1) -- (1.02,1);
  \draw[thin, dashed, color=gray] (-0.02,1+2*\dd) -- (1.02,1+2*\dd); 

    \draw [fill=gray, opacity=0.4]
       (0, 3*\dd) -- (0.5,1+2*\dd)  -- (1,2*\dd) -- (1,0) -- (0.5, 1) -- (0, \dd) -- cycle;
\end{tikzpicture}
\quad
\begin{tikzpicture}[scale=0.77,xscale=5,yscale=4,domain=0:1]
  \draw[->] (-0.02,0) -- (1.05,0) node[below right] {$x_2$};
  \draw[->] (0,-0.02) -- (0,1.08+2*\dd) node[above] {$f(x_2)$};

  \draw (0,0.02) -- (0,-0.02) node[below] {$0$};
  \draw (0.5,0.02) -- (0.5,-0.02) node[below] {$\frac{1}{2}$};
  \draw (1,0.02) -- (1,-0.02) node[below] {$1$};

  \draw (0.98,2*\dd) -- (1.02, 2*\dd) node[right]{$2\delta$};
  \draw (0.02,3*\dd) -- (-0.02, 3*\dd) node[left]{$3\delta$};
  \draw (0.02,1+2*\dd) -- (-0.02, 1+2*\dd) node[left]{$1+2\delta$};

  \draw[color=blue,ultra thick]   (0,\dd+2*\dd) -- (0.5,1+2*\dd);
  \draw[color=blue, ultra thick]   (0.5,1+2*\dd) -- (1,2*\dd);

  \draw[thin, dashed, color=gray] (0.5,-0.02) -- (0.5,1.02+2*\dd);
  \draw[thin, dashed, color=gray] (1,-0.02) -- (1,1.02+2*\dd);
  \draw[thin, dashed, color=gray] (-0.02,1+2*\dd) -- (1.02,1+2*\dd); 
\end{tikzpicture}
\caption{An instance with supra-positive investment returns.}\label{fig:supra-improvement}
\end{figure}

For a second illustration, let's return to our running example.
In Figure~\ref{fig:supra-improvement}, the solid lines indicate the
starting utility functions. These functions correspond to the 
terminal functions in Figure~\ref{fig:supra-negative} which we have seen give an equilibrium with social welfare $2\delta$, despite ${\tt opt}=1+\delta$. Suppose we now 
make a targeted investment in the red neighbourhood, thus improving its
utility function to the dashed red lines. The cost of this improvement is $2\delta$, the area of the shaded segment in Figure~\ref{fig:supra-improvement}, but it immediately leads to a unique Nash equilibrium with
social welfare $1+2\delta$, up from the original equilibrium social welfare of $2\delta$. Thus after a very small investment, the micro-behaviour dynamics induce an outcome with social welfare even greater than
the original optimal social welfare.

Again it is important to emphasise how surprising the result in Theorem~\ref{thm:main} is.
The alternative approach to allow the agents to have a total utility of
$\epsilon\cdot {\tt opt}$ would be via a subsidy or tax reductions.
But the cost of these approaches is typically 
$\Omega(\epsilon\cdot {\tt opt})$, a dramatic cost increase over $O(\epsilon^2\cdot {\tt opt})$.

The reader may ask whether Theorem~\ref{thm:main} can be improved: {\em is it possible to obtain such supra-positive returns for a lower cost?}
The answer is no. More specifically, while it may be possible to obtain minor improvements 
in the constant, the dominant multiplicative factor in the cost function, that is $\epsilon^2$, is tight and cannot be improved. We prove this in Section~\ref{sec:lower} with the following theorem.

\begin{theorem}\label{thm:minor}
There exist urban planning instances
for which ensuring every resulting equilibrium has welfare at least $\epsilon\cdot {\tt opt}$ requires targeted investments with a total cost at least $0.12\epsilon^2\cdot {\tt opt}$.
\end{theorem}

\section{Background and Related Literature}\label{sec:background}

Our study builds on a rich body of work at the intersection of game theory, urban planning, and agent-based modelling. A central theme in this literature is the disconnect between individual incentives and collective welfare -- a phenomenon famously explored by Schelling in his seminal model of segregation dynamics~\cite{schelling1971dynamic}. Schelling showed that even slight individual preferences for similar neighbours can drive highly segregated group outcomes. 

While Schelling’s original work focused on a distance based neighbourhood model in general graphs, he also introduced the \emph{bounded neighbourhood model}, where agents evaluate entire neighbourhoods rather than immediate surroundings~\cite{schelling1971dynamic}. In this formulation, a city is partitioned into blocks with finite capacities, and agents decide to enter or exit based on the overall demographic or density profile of each block. The bounded neighbourhood model revealed how strategic location choices, even under mild individual preferences, can generate sharp macro-level discontinuities, such as tipping points \cite{CMR08, Grod57, Zhang11} and abrupt shifts in neighbourhood composition.

A number of subsequent papers build on the bounded neighbourhood framework, showing how modest levels of altruism \cite{grauwin2009competition, jensen2018giant} and tolerance \cite{gauvin2009phase} can bring the system closer to social optimum. These results reinforce the insight that local incentives can misalign sharply with global welfare, a phenomenon we aim to mitigate through targeted intervention.

Su et al.~\cite{su2020intriguing,su2023significant} restrict how individuals can move. In particular, they found that even purely egoistic behaviour can yield efficient equilibria if individuals can only move between blocks via a star-like connection topology. These results suggest that structural and behavioural interventions alike can help mitigate inefficiencies.

To contextualize our work within the broader economic literature, we highlight recent studies of modified Schelling’s model with game-theoretic approaches. Grauwin et al.~\cite{grauwin2012dynamic} analysed the model using evolutionary game theory, showing that linear utility functions align individual and collective welfare, whereas Schelling’s original and asymmetric peaked utilities often lead to stable but suboptimal segregated outcomes. 

Grauwin et al.~\cite{grauwin2009competition} slightly depart from Schelling's model with two populations and propose a simplified version with homogeneous agents and utility functions that depend only on local density. Methodologically, our work is closest to that of Grauwin et al.~\cite{grauwin2009competition}, though we shift the focus from emergent behaviour to neighbourhood design. Rather than assume agents change behaviour intrinsically, we ask how targeted investments can promote socially desirable equilibria. 

More recently, Schelling's original model has also been studied from a game theoretic perspective~\cite{chauhan2018schelling,echzell2019convergence}. In this model, agents are happy if at least a certain fraction of neighbours are of the same type. There the authors provide results on the convergence to equilibria and on the price of anarchy, where social utility is defined as the number of happy individuals. Because of the different definition, those results cannot be compared to our total utility based results. Kanellopoulos et al.~\cite{kanellopoulos2021modified} and Bil\`o et al.~\cite{bilo2022topological} further examined equilibrium inefficiencies under varying dynamics. Agarwal et al. \cite{agarwal2020swap, agarwal2021schelling} study a generalisation of this model with $k$ rather than just $2$ types of agents.

Additionally, our model shares similarities with congestion games, a well-established area in economics and computer science~\cite{rosenthal1973class,roughgarden2005selfish,roughgarden2002bad}. 
In the congestion game literature, the price of anarchy has been extensively studied for both the non-atomic \cite{roughgarden2002bad} and atomic \cite{AlandDGMS11, BhawalkarGR14, roughgarden2015intrinsic} versions, and interventions to improve the improve the price of anarchy (see \cite{PaccagnanG24} and references therein). 
However, while traditional congestion game studies assume monotonic utility structures for agents, we adopt a concave utility framework. This modification better aligns with the dynamics of the Schelling-inspired family of games and is necessary in allowing us to capture the more nuanced trade-offs involved in strategic decision-making and resource allocation for urban neighbourhood games. This new perspective allows us to better understand and mitigate inefficiencies in strategic agent interactions, for example, by exploiting dynamic and tipping behaviours.

\section{Case Study: Reversing the Donut Effect}\label{sec:case-study}

A classical problem in urban planning is deurbanisation, whereby population and economic activity shift from dense urban cores toward surrounding suburban areas~\cite{baum2007did,henderson2004handbook,mieszkowski1993causes}.
Historically, this dynamic has been studied in the context of race~\cite{Grod57}. Indeed, during the second half of the twentieth century, this process was observed across many U.S. metropolitan areas, including Detroit, St. Louis, Cleveland, Baltimore, and Philadelphia, where suburban population growth coincided with stagnation or decline in inner-city populations. 
More recently, the dynamic has been viewed with respect to 
differences in the quality of services between urban and suburban
areas, traffic housing costs, traffic and pollution, and the
popularity of remote work. The resultant negative dynamics leading to hollowing out of the urban core has been dubbed the ``donut effect''~\cite{ramani2021donut}.

In this section, we explain how targeted investments can 
be used to reverse this negative dynamic. 
An illustration of the spatial pattern corresponding to the donut effect using our model is given in Figure~\ref{fig:city_initial}. 
In this example $n=11$, with one large central urban core and ten smaller suburban regions. Originally, the utility functions in the urban area are similar to that of the suburban areas. But, despite that, a negative dynamic forms whereby population moves from the inner-city to the suburbs. The process terminates when the suburban areas reach a population density of $84\%$ (represented by the blue blocks) and the inner-city has a 
population density of just $16\%$ (represented by the red blocks).
For clarity of discussion, the technical details leading to these statistics are not important here; we refer the reader to the appendix for the associated utility functions and formal mathematical details.


\begin{figure}[h]
    \centering
    \begin{subfigure}{0.45\textwidth}
        \centering
        \includegraphics[width=\linewidth]{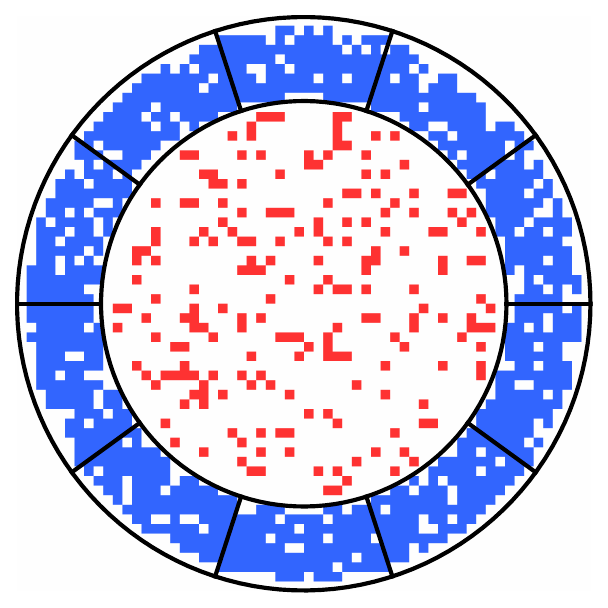}
        \caption{Initially, the equilibrium configuration has the suburban region populated with density $\rho_{B}=0.84$, while the inner-city region has low population density $\rho_{R}=0.16$.}
        \label{fig:city_initial}
    \end{subfigure}
    \hfill
    \begin{subfigure}{0.45\textwidth}
        \centering
        \includegraphics[width=\linewidth]{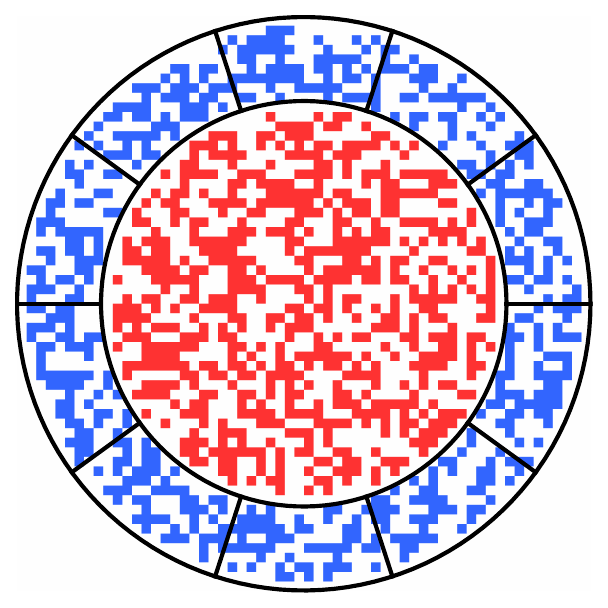}
        \caption{After a carefully targeted investment, the equilibrium configuration has a balanced population density of $\rho = 0.5$in in both the suburban neighbourhoods and the inner-city region.}
        \label{fig:city_improved}
    \end{subfigure}

    \caption{A circular city composed of one central inner-city region (red) and an equal-area suburban region (blue), the latter subdivided into ten equally sized neighbourhoods. The spatial domain is discretized on a $60 \times 60$ grid, and population density is represented by coloured filled pixels.}
    \label{fig:circular_city}
\end{figure}

However, it is possible to reverse this dynamic using small targeted investments. Indeed, small investments to slightly improve the utility function in the inner-city causes the population to iteratively
move back into the urban core. This positive dynamic terminates
when the inner-city and each suburban region have a population density of $50\%$. This is illustrated in Figure~\ref{fig:city_improved}. We remark that the terms negative and positive dynamic are appropriate here because every agent is better off after the targeted investment (again the details are given in the appendix).

This example illustrates, therefore, that our model qualitatively captures a spatial configuration that reflects a robust and widely documented occurrence in urban development. It also highlights ways in which the model illuminates ways to circumvent
negative dynamics that arise in urban environments using 
targeted investments.

\section{Carefully Targeted Investments Guarantee High Welfare}\label{sec:upper}
In this section we prove our main result, Theorem~\ref{thm:main}. 
\begin{restatetheorem}{thm:main}
 In any urban planning instance, there exist targeted investments with a total cost at most $0.81 \epsilon^2\cdot {\tt opt}$ that ensure every resulting equilibrium has welfare at least $\epsilon\cdot {\tt opt}$.   
\end{restatetheorem}
That is, it is always possible to make low cost targeted investments to ensure any resultant equilibrium has high social welfare. To achieve this, there are two fundamental issues: how and where?
The ``where'' concerns in which neighbourhoods should we make our targeted investments? 
The ``how'' concerns what forms these targeted investments should take in the chosen neighbourhoods. We design an algorithm that addresses both questions. In Section~\ref{sec:how}, we present one method that will prove fruitful for the ``how''. In Section~\ref{sec:where}, we describe an effective method to select which neighbourhoods to invest it.

Before giving a detailed description of the algorithm and its mathematical analysis, let's recall the model and a formal description of the problem. Recall there are $n$ locations with valuation functions $f_1,f_2,\dots, f_n$ each on $[0,u_i]$, where $u_i$ is the capacity of location $i$. 
Let there be $N$ agents where we may assume that $\max_{i=1}^n u_i \le N \le \sum_{i=1}^n u_i$.
To see this, one, in practice $N$ is a defacto upper bound on $u_i$ as location $i$ can never contain more agents than the entire population. 
Two, in practice
$N$ is a defacto lower bound on $\sum_{i=1}^n u_i$
as together all the locations must be able to hold the
entire population. 

Furthermore, recall each function $f_i$ is non-negative and concave, but importantly need not be monotonic. 

Our objective is to obtain functions $g_1,g_2,\dots, g_n$ such that
\begin{enumerate}
\item The social welfare at equilibrium for $\{g_1,g_2,\dots, g_n\}$ is at least $\epsilon\cdot {\tt opt}$.
\item The improvement cost is small. Specifically,
the symmetric differences satisfy $\sum_{i=1}^n |f_i \oplus g_i| \le \epsilon^2\cdot {\tt opt}$.
\end{enumerate}
In fact, we show how to obtain a solution of utility at least $\epsilon \cdot {\tt opt}$ by incurring a cost of at most just $0.81\epsilon^2 \cdot {\tt opt}$.

\subsection{How: A Strategy for Improving a Location}\label{sec:how}

Suppose we select a location in which to make a targeted investment. What form should these investments take? For the purpose of obtaining
a resultant equilibrium with high social welfare, the
strategy we apply is to focus on the extreme cases, where the neighbourhood is either very sparsely populated or very densely populated. Our aim is that, after the investment, the neighbourhood will proffer its residents a reasonable (albeit, possibly quite small) utility when it is sparsely or densely populated. 

Let's now formally define this investment strategy,
which we dub the {\em rudimentary} strategy,
on location $i$. Assume $f_i$ has maximum value $h_i$
and passes through the three points
$(0,\ell_i)$, $(\sigma_i, h_i)$ and $(u_i, \ell'_i)$. (Possibly the middle point is the same as one of the two end points.)

The rudimentary strategy to improve a location $i$ is then as follows. Assign a {\em target payoff} of $\tau_i=\epsilon\cdot h_i$ for location $i$.
If $\ell_i< \tau_i$ then let $\alpha_i$ satisfy $f(\alpha_i)= \tau_i$ where $\alpha_i \le \sigma_i$; otherwise set $\alpha_i=0$. 
Similarly, if $\ell'_i< \tau_i$ then let $\beta_i$ satisfy $f(\beta_i)= \tau_i$ where $\beta_i \le \sigma_i$; otherwise set $\beta_i=u_i$. 

Given $\alpha_i$ and $\beta_i$, define a new function $g_i$ as
$$
g_i(x) =
\begin{cases}
\tau_i &\text{if } 0 \le  x < \alpha_i \\
f_i(x) &\text{if } \alpha_i \le x \le \beta_i\\
\tau_i &\text{if } \beta_i < x \le u_i
\end{cases}
$$
We can upper bound the total cost of upgrading $f_i$ to $g_i$ via the following lemma.
\begin{lemma}\label{lem:cost}
The cost of the rundimentary strategy applied at location $i$, with target $\tau_i=\epsilon\cdot h_i$, is at most $\frac{\epsilon^2}{2} \cdot u_i\cdot h_i$.
\end{lemma}
\begin{proof}
Let's assess the cost of improving $f_i$ to $g_i$.
By concavity the line segment between $(0, \ell_i)$ and $(\sigma_i, h_i)$ lies below $f_i$. This line segment has slope 
$\frac{(h_i-\ell_i)}{\sigma_i}$.
It follows that $\alpha_i \le (\tau_i-\ell_i)\cdot \frac{\sigma_i}{(h_i-\ell_i)}$.

The cost of improving $f_i$ to $g_i$ in the range $[0,\alpha_i]$ is then at most
\begin{eqnarray}\label{eq:improve-low}
\frac12\cdot \alpha_i\cdot (\tau_i-\ell_i) \le
\frac{\sigma_i}{2(h_i-\ell_i)}\cdot(\tau_i-\ell_i)^2.
\end{eqnarray}
Symmetrically, by concavity, the line segment between $(\sigma_i, h_i)$ and $(u_i, \ell'_i)$ lies below $f_i$.
This line segment has slope $-\frac{(h_i-\ell'_i)}{(u_i-\sigma_i)}$.
It follows that $u_i-\beta_i \le (\tau_i-\ell'_i)\cdot \frac{(u_i-\sigma_i)}{(h_i-\ell'_i)}$.

The cost of improving $f_i$ to $g_i$ in the range $[\beta_i, u_i]$ is then at most
\begin{eqnarray}\label{eq:improve-high}
\frac12\cdot (u_i-\beta_i)\cdot (\tau_i-\ell'_i) 
\le
\frac{(u_i-\sigma_i)}{2(h_i-\ell'_i)}\cdot(\tau_i-\ell'_i)^2.
\end{eqnarray}
Therefore, by (\ref{eq:improve-low}) and (\ref{eq:improve-high}), the total cost making these two improvements on locations $i$ is at most
\begin{eqnarray}\label{eq:improve-both}
\frac{\sigma_i}{2(h_i-\ell_i)}\cdot(\tau_i-\ell_i)^2 +
\frac{(u_i-\sigma_i)}{2(h_i-\ell'_i)}\cdot(\tau_i-\ell'_i)^2 
&\le&
\frac{\sigma_i}{2h_i}\cdot \tau_i^2 +
\frac{(u_i-\sigma_i)}{2h_i}\cdot \tau_i^2 \nonumber\\ 
&=& \frac{u_i\cdot \tau_i^2}{2 h_i} \nonumber\\ 
&=& \frac{\epsilon^2}{2} \cdot u_i\cdot h_i.
\end{eqnarray}
Here the inequality holds because $\tau_i\le h_i$, and the second equality holds as $\tau_i=\epsilon \cdot h_i$.
\end{proof}

\subsection{Where: Selecting Locations for Improvement}\label{sec:where}

We now address the complex question of where to make our targeted investments.
The problem is that it will be too expensive to attempt to improve every location.
Thus, we must be judicious in our choice of which locations to improve using the rudimentary strategy.

To select locations, our algorithm breaks the problem into two cases. We say that a location $\ell^*$ is {\em critical} if $u_{\ell^*}\cdot h_{\ell^*} \ge \phi\cdot {\tt opt}$, where $\phi=\frac{\sqrt{5}-1}{2}\approx 0.618$ is the golden ratio conjugate.
Either there is a critical location in the urban planning instance or no critical locations exist.

In the unlikely situation where a critical location exist then the selection problem is straightforward.
We must target our investments in the critical location. Let's formalize this approach and analyze its cost and performance.\\

\noindent (1) There is a critical location $\ell^*$. 

Let $u_{\ell^*}\cdot h_{\ell^*} = \lambda\cdot {\tt opt}$, where $\lambda\ge \phi$.
In this case, our strategy is simple: only improve location $\ell^*$.
In particular, we apply the rudimentary strategy at location $\ell^*$ using a target payoff
of $\tau_{\ell^*}= \epsilon^*\cdot h_{\ell^*}$, where $\epsilon^*= \frac{1}{\lambda}\cdot \epsilon$.

\begin{lemma}\label{lem:case1}
If there is a critical location then the total improvement cost is at most $0.81 \cdot \epsilon^2 \cdot {\tt opt}$ and the social welfare of
any resultant equilibrium is at least $\epsilon \cdot {\tt opt}$.
\end{lemma}
\begin{proof}
Take any Nash equilibrium ${\bf x}$ after applying the
rudimentary strategy at location $\ell^*$ using a target payoff of $\tau_{\ell^*}= \epsilon^*\cdot h_{\ell^*}$.
If $x_{\ell^*} = u_{\ell^*}$ then the resultant social welfare is at least
\begin{equation}\label{eq:case-1a}
u_{\ell^*}\cdot \tau_{\ell^*} 
\ =\  \epsilon^* \cdot u_{\ell^*}\cdot h_{\ell^*} 
\ =\  \frac{1}{\lambda}\cdot \epsilon \cdot u_{\ell^*}\cdot h_{\ell^*} 
\ =\  \epsilon \cdot {\tt opt}.
\end{equation}
On the other hand assume $x_{\ell^*} < u_{\ell^*}$.
Then by the Nash equilibrium conditions every agent has utility
at least $\tau_{\ell^*}$ otherwise it would improve its utility by moving to location $\ell^*$.
Thus the social welfare is at least
\begin{equation}\label{eq:case-1b}
N\cdot \tau_{\ell^*} 
\ =\ N\cdot \epsilon^* \cdot h_{\ell^*} 
\ =\ N\cdot \frac{1}{\lambda}\cdot \epsilon \cdot h_{\ell^*} 
\ \ge\ \frac{1}{\lambda}\cdot \epsilon \cdot u_{\ell^*}\cdot h_{\ell^*} 
\ =\ \epsilon \cdot {\tt opt}.
\end{equation}
Above, the inequality holds because $N\ge  \max_{i=1}^n u_i \ge u_{\ell^*}$.

It follows by (\ref{eq:case-1a}) and (\ref{eq:case-1b}) that the Nash
equilibrium has social welfare at least $\epsilon \cdot {\tt opt}$.
What is the cost incurred in obtaining this welfare guarantee?
By Lemma~\ref{lem:cost}, the cost incurred is at most
\begin{eqnarray}
\frac{(\epsilon^*)^2}{2} \cdot u_{\ell^*}\cdot h_{\ell^*}
&=& \frac{\epsilon^2}{2\lambda^2} \cdot u_{\ell^*}\cdot h_{\ell^*} \nonumber\\
&=& \frac{\epsilon^2}{2\lambda} \cdot {\tt opt} \nonumber\\
&\le& \frac{\epsilon^2}{2\phi} \cdot {\tt opt}\nonumber\\
&=& \frac{1+\phi}{2}\cdot \epsilon^2 \cdot {\tt opt}.
\end{eqnarray}
Here the final inequality holds as $\phi$ is the golden ratio conjugate and, thus, $\phi(1+\phi)=1$.

Thus we obtain a solution of utility at least $\epsilon \cdot {\tt opt}$ for a cost of just at most 
$\frac{1+\phi}{2}\cdot \epsilon^2 \cdot {\tt opt}\approx 0.81 \cdot \epsilon^2 \cdot {\tt opt}$.
\end{proof}

The selection of which locations to improve 
is much more subtle for the case in which there are no critical locations.

\noindent (2) There are no critical locations.

It follows that every location $i$ has the property that $u_{i}\cdot h_{i} < \phi\cdot {\tt opt}$.
Order the locations such that $h_1\ge h_2\ge \cdots \ge h_n$.
Let 
$$k^* = \min_k \  \{k: \sum_{i=1}^k u_i\cdot h_i \ge {\tt opt}\}.$$
Observe that $k^*$ exists as ${\tt opt} \le \sum_{i=1}^n u_i\cdot h_i$.
Let $\lambda$ be such that 
\begin{equation}\label{eq:lambda}
\sum_{i=1}^{k^*} u_i\cdot h_i  \ = \  \lambda\cdot {\tt opt}.
\end{equation}
Observe that $1\le \lambda$ and 
\begin{align}\label{eq:lambda-upper}
    \lambda< 1+\phi,
\end{align} which follows as 
$\sum_{i=1}^{k^*-1} u_i\cdot h_i < {\tt opt}$ and, by case assumption, $u_{k^*}\cdot h_{k^*} < \phi\cdot {\tt opt}$.

Our strategy now is to only improve locations $1$ through to $k^*$. No improvements are made in
locations $k^*+1$ through to $n$.
Specifically, we apply the rudimentary strategy
with a target of $\tau_i=\epsilon \cdot h_i$ on 
location $i$, for each $1\le i \le k^*$. 

\begin{lemma}\label{lem:case2}
If there are no critical location then the total improvement cost is at most $0.81 \cdot \epsilon^2 \cdot {\tt opt}$ and the social welfare of
any resultant equilibrium is at least $\epsilon \cdot {\tt opt}$.
\end{lemma}
\begin{proof}
Let ${\bf x}$ be a resultant Nash equilibrium.
Let $j^* \le k^*$ be the first location such that $x_{j^*} < u_{j^*}$, if such a location exists. 

If no such location exists then set $j^*=k^*+1$.
Because locations $1$ to $k^*$ are full,
the total utility of the Nash equilibrium is at least
\begin{eqnarray}\label{eq:case-2a}
\sum_{i=1}^{k^*} u_i\cdot \tau_i 
&=& \epsilon \cdot \sum_{i=1}^{k^*} u_i\cdot h_i\nonumber \\
&\ge&  \epsilon \cdot {\tt opt} .
\end{eqnarray}

On the other hand, assume that $j^* \le k^*$.
Then at this Nash equilibrium, every agent at a location $i> j^*$
has a utility of at least at least $\tau_{j^*}$, otherwise it would
move to location $j^*$.

Again, because locations $1$ to $j^*-1$ are full, the total utility of the Nash equilibrium is then at least
\begin{eqnarray}\label{eq:case-2b}
\sum_{i=1}^{j^*-1} u_i\cdot \tau_i + (N-\sum_{i=1}^{j^*-1} u_i)\cdot \tau_{j^*}
&=& 
\epsilon \cdot \left( \sum_{i=1}^{j^*-1} u_i\cdot h_i + (N-\sum_{i=1}^{j^*-1} u_i)\cdot h_{j^*}\right) \nonumber\\
&\ge& 
\epsilon \cdot \left( \sum_{i=1}^{k^*-1} u_i\cdot h_i + (N-\sum_{i=1}^{k^*-1} u_i)\cdot h_{k^*}\right)\nonumber\\
&\ge& \epsilon\cdot {\tt opt} .
\end{eqnarray}
Above, the first inequality holds because the $h_i$ are in decreasing order. 
Furthermore, since the $h_i$ are decreasing, it must be the case that
${\tt opt}\le \sum_{i=1}^{k^*-1} u_i\cdot h_i + (N-\sum_{i=1}^{k^*-1} u_i)\cdot h_{k^*}$, giving the second inequality.

It follows by (\ref{eq:case-2a}) and (\ref{eq:case-2b}) that the Nash
equilibrium has social welfare at least $\epsilon \cdot {\tt opt}$.
What is the cost incurred in obtaining this welfare guarantee?
Recall we are only making improvements for locations $1$ to $k^*$
(i.e. using using $g_i$ instead of $f_i$, for $1\le i \le k^*$).
For locations $k^*+1$ to $n$ no improvements are made (i.e. we continue to use $f_i$, for $k^*+1\le i \le n$) and
thus no costs incurred on those locations.

Thus, by (\ref{eq:improve-both}),  the cost of our improvements on locations $1$ to $k^*$ is at most
\begin{eqnarray}\label{eq:cost}
\sum_{i=1}^{k^*} \frac{\epsilon^2}{2}\cdot u_i\cdot h_i 
&=& \frac{\epsilon^2}{2}\cdot \sum_{i=1}^{k^*} u_i\cdot h_i 
\nonumber\\
&=& \frac{\epsilon^2}{2} \cdot \lambda\cdot{\tt opt} \nonumber \\
&\le& \frac{1+\phi}{2}\cdot \epsilon^2 \cdot {\tt opt} .
\end{eqnarray}
Here the equality follows by \eqref{eq:lambda} and the inequality by \eqref{eq:lambda-upper}.

Thus, again, we obtain a solution of utility at least $\epsilon \cdot {\tt opt}$ for a cost of just at most 
$\frac{1+\phi}{2}\cdot \epsilon^2 \cdot {\tt opt}\approx 0.81 \cdot \epsilon^2 \cdot {\tt opt}$.
\end{proof}
This leads to our main result.
\begin{proof}[Proof of Theorem~\ref{thm:main}]
The theorem is an immediate consequence of Lemma~\ref{lem:case1} and
Lemma~\ref{lem:case2}.
\end{proof}

\section{An Optimal Guarantee}\label{sec:lower}

The reader may ask whether the guarantee provided by Theorem~\ref{thm:main} can be improved. The answer is no:
the theorem gives the strongest achievable
guarantee. Specifically, if we wish to ensure every resultant equilibrium has social welfare at least $\epsilon\cdot {\tt opt}$ then it is necessary to make investments of total cost $\Omega(\epsilon^2\cdot {\tt opt})$. Thus, the $\epsilon^2$ term is necessary and Theorem~\ref{thm:main} is tight to within a
constant factor. This is specified in the following theorem.
\begin{restatetheorem}{thm:minor}
There exist urban planning instances
for which ensuring every resulting equilibrium with welfare at least $\epsilon\cdot {\tt opt}$ requires targeted investments of total cost at least $0.12\epsilon^2\cdot {\tt opt}$.
\end{restatetheorem}
\begin{proof}
To prove the theorem it suffices to present one example with this property. To do this, we construct an instance with $n=2$ locations, each with capacities $u_1=u_2=1$, and where the number of agents is $N=1$. Let $f_1$ be a piecewise linear function through the two points
$(0,1)$ and $(1,0)$. Thus $h_1=1$.
Let $f_2$ be a constant function through the two points
$(0,0)$ and $(1,0)$. Thus $h_2=0$.
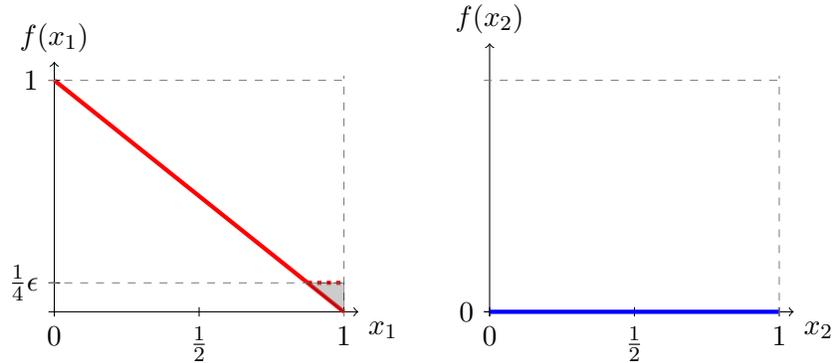
\begin{figure}[h!]
\centering
\newcommand{\dd}{0.04}
\newcommand{\ee}{0.5}
\begin{tikzpicture}[scale=0.77,xscale=5,yscale=4,domain=0:1]
  \draw[->] (-0.02,0) -- (1.05,0) node[below right] {$x_1$};
  \draw[->] (0,-0.02) -- (0,1.08) node[above] {$f(x_1)$};

  \draw (0,0.02) -- (0,-0.02) node[below] {$0$};
  \draw (0.5,0.02) -- (0.5,-0.02) node[below] {$\frac{1}{2}$};
  \draw (1,0.02) -- (1,-0.02) node[below] {$1$};
  
  \draw[color=red, ultra thick]   (0,1) -- (1,0);
  \draw (0.02,1) -- (-0.02, 1) node[left]{$1$};

  \draw[color=red,ultra thick,dotted]   (1-\ee/4,\ee/4) -- (1,\ee/4);
  \draw (0.02,\ee/4) -- (-0.02, \ee/4) node[left]{$\frac14\epsilon$};

  \draw[thin, dashed, color=gray] (1,-0.02) -- (1,1.02);
  \draw[thin, dashed, color=gray] (-0.02,1.0) -- (1.02,1.0);
  \draw[thin, dashed, color=gray] (-0.02,\ee/4) -- (1.02,\ee/4); 

    \draw [fill=gray, opacity=0.4]
       (1-\ee/4,\ee/4) -- (1,\ee/4) -- (1,0) -- cycle;
\end{tikzpicture}
\quad
\begin{tikzpicture}[scale=0.77,xscale=5,yscale=4,domain=0:1]
  \draw[->] (-0.02,0) -- (1.05,0) node[below right] {$x_2$};
  \draw[->] (0,-0.02) -- (0,1.08+2*\dd) node[above] {$f(x_2)$};

  \draw (0,0.02) -- (0,-0.02) node[below] {$0$};
  \draw (0.5,0.02) -- (0.5,-0.02) node[below] {$\frac{1}{2}$};
  \draw (1,0.02) -- (1,-0.02) node[below] {$1$};

  \draw (0.02,0) -- (-0.02, 0) node[left]{$0$};

  \draw[color=blue,ultra thick]   (0,0) -- (1,0);

  \draw[thin, dashed, color=gray] (1,-0.02) -- (1,1.02);
  \draw[thin, dashed, color=gray] (-0.02,1) -- (1.02,1); 
\end{tikzpicture}
\caption{Lower bound example}\label{fig:lower-bound}
\end{figure}

Setting $x_1=\lambda$ and $x_2=1-\lambda$ gives a social welfare
of $\lambda(1-\lambda)$ which is maximized when $\lambda=\frac12$.
Thus ${\tt opt}=\frac14$.

Now assume we want to create a Nash equilibrium with social welfare
$\epsilon\cdot {\tt opt}= \frac14 \epsilon$.
To do this, the reader may observe that it is more expensive to improve location $2$ than location $1$. So the cheapest solution is to only improve location $1$. To do this we set:
$$
g_1(x) =
\begin{cases}
f_i(x) &\text{if } 0 \le x \le 1-\frac14 \epsilon\\
\frac14 \epsilon &\text{if } 1-\frac14 \epsilon < x \le 1
\end{cases}
$$

Observe there is now a unique Nash equilibrium in which every agent moves to location $1$ and each agent receives a payoff of
$\frac14 \epsilon$. Since there are $N=1$ units of of agents
this gives a social welfare of $\epsilon\cdot {\tt opt}$ as desired.

The total cost of improving $f_1$ to $g_1$ is 
$$\frac12 \left(\frac14 \epsilon\right)^2 = \frac{1}{32}\epsilon^2 = \frac{1}{8}\epsilon^2\cdot {\tt opt}.$$

Hence the optimal solution has minimum cost $\Omega(\epsilon^2\cdot {\tt opt})$.
\end{proof}

Observe that the multiplicative constants in Theorems~\ref{thm:main} and~\ref{thm:minor} are 
$0.81$ and $\frac18$, respectively. It is an interesting open problem to reduce the gap between these upper and lower bounds. More optimistically, is it possible to close the gap completely to obtain the exact multiplicative constant for our supra-positive results?

\section{Conclusion}\label{sec:conc}

Using a generalized version of Schelling's {\em bounded neighbourhood model}~\cite{schelling1971dynamic}, we
proved that either the current neighbourhood plan already proffers social welfare comparable to the optimal social welfare achievable or there exist
targeted neighbourhood investments that produce supra-positive returns. This result serves as a theoretic ``proof of concept''. 
This immediately induces two critical questions.
First, does this result extend to the more realistic models of urban planning? Second, can this result be verified in experimental settings?

\bibliography{bib}

@article{agarwal2021schelling,
  title={Schelling games on graphs},
  author={Agarwal, A. and Elkind, E. and Gan, J. and Igarashi, A. and Suksompong, W. and Voudouris, A.},
  journal={Artificial Intelligence},
  volume={301},
  pages={103576},
  year={2021},
  publisher={Elsevier}
}

@inproceedings{agarwal2020swap,
  title={Swap stability in {S}chelling games on graphs},
  author={Agarwal, A. and Elkind, E. and Gan, J. and Voudouris, A.},
  booktitle={Proceedings of the 34th Conference on Artificial Intelligence (AAAI)},
  pages={1758--1765},
  year={2020}
}

@article{AlandDGMS11,
  author       = {Aland, S and
                  Dumrauf, D. and
                  Gairing, M. and
                  Monien, B. and
                  Schoppmann, F.},
  title        = {Exact Price of Anarchy for Polynomial Congestion Games},
  journal      = {{SIAM} Journal on Computing},
  volume       = {40},
  number       = {5},
  pages        = {1211--1233},
  year         = {2011},
  url          = {https://doi.org/10.1137/090748986},
  doi          = {10.1137/090748986},
  bibsource    = {dblp computer science bibliography, https://dblp.org}
}

@article{baum2007did,
  title={Did highways cause suburbanization?},
  author={Baum-Snow, N.},
  journal={The Quarterly Journal of Economics},
  volume={122},
  number={2},
  pages={775--805},
  year={2007},
  publisher={MIT Press}
}

@article{BhawalkarGR14,
  author       = {Bhawalkar, K. and
                  Gairing, M. and
                  Roughgarden, T.},
  title        = {Weighted Congestion Games: The Price of Anarchy, Universal Worst-Case
                  Examples, and Tightness},
  journal      = {{ACM} Transactions on Economics and Computing},
  volume       = {2},
  number       = {4},
  pages        = {14:1--14:23},
  year         = {2014},
  url          = {https://doi.org/10.1145/2629666},
  doi          = {10.1145/2629666},
  bibsource    = {dblp computer science bibliography, https://dblp.org}
}

@article{BT08,
  title={Segregation and strategic neighborhood interaction},
  author={Barr, J. and Tassier, T.},
  journal={Eastern Economic Journal},
  volume={34},
  pages={408--503},
  year={2008},
}

@article{Bra68,
  title={Uber ein Paradoxon aus der Verkehrsplanung},
  author={Braess, D.},
  journal={Unternehmensforschung Operations Research - Recherche Opérationnelle},
  volume={12},
number={1},
  pages={258--268},
  year={1968},
}

@article{bilo2022topological,
  title={Topological influence and locality in swap {S}chelling games},
  author={Bil{\`o}, D. and Bil{\`o}, V. and Lenzner, P. and Molitor, L.},
  journal={Autonomous Agents and Multi-Agent Systems},
  volume={36},
  number={2},
  pages={47},
  year={2022},
  publisher={Springer}
}

@article{CMR08,
  title={Tipping and the dynamics of segregation},
  author={Card, D. and Mas, A. and Rothstein, J.},
  journal={Quarterly Journal of Economics},
  volume={123},
  pages={177--218},
  year={2008},
}

@inproceedings{chauhan2018schelling,
  title={Schelling segregation with strategic agents},
  author={Chauhan, A. and Lenzner, P. and Molitor, L.},
  booktitle={Proceedings of 11th International Symposium on Algorithmic Game Theory (SAGT)},
  pages={137--149},
  year={2018},
  NOOPorganization={Springer}
}

@inproceedings{echzell2019convergence,
  title={Convergence and hardness of strategic {S}chelling segregation},
  author={Echzell, H. and Friedrich, T. and Lenzner, P. and Molitor, L. and Pappik, M. and Sch{\"o}ne, F. and Sommer, F. and Stangl, D.},
  booktitle={Proceedings of 15th International Conference on Web and Internet Economics (WINE)},
  pages={156--170},
  year={2019},
  NOOPorganization={Springer}
}

@article{gauvin2009phase,
  title={Phase diagram of a {S}chelling segregation model},
  author={Gauvin, L. and Vannimenus, J. and Nadal, J-P.},
  journal={The European Physical Journal B},
  volume={70},
  pages={293--304},
  year={2009},
  publisher={Springer}
}

@article{grauwin2009competition,
  title={Competition between collective and individual dynamics},
  author={Grauwin, S. and Bertin, E. and Lemoy, R. and Jensen, P.},
  journal={Proceedings of the National Academy of Sciences},
  volume={106},
  number={49},
  pages={20622--20626},
  year={2009},
  publisher={National Acad Sciences}
}

@article{grauwin2012dynamic,
  title={Dynamic models of residential segregation: An analytical solution},
  author={Grauwin, S. and Goffette-Nagot, F. and Jensen, P.},
  journal={Journal of Public Economics},
  volume={96},
  number={1-2},
  pages={124--141},
  year={2012},
  publisher={Elsevier}
}

@book{Grod57,
  title={Metropolitan Segregation},
  author={Grodzins, M.},  
year={1957},
  publisher={University of Chicago Press},
}

@book{henderson2004handbook,
  title={Handbook of Regional and Urban Economics: Cities and Geography},
  author={Henderson, V. and Thisse, J-F.},
  volume={4},
  year={2004},
  publisher={Elsevier}
}

@article{jensen2018giant,
  title={Giant catalytic effect of altruists in {S}chelling’s segregation model},
  author={Jensen, P. and Matreux, T. and Cambe, J. and Larralde, H. and Bertin, E.},
  journal={Physical Review Letters},
  volume={120},
  number={20},
  pages={208301},
  year={2018},
  publisher={APS}
}

@article{kanellopoulos2021modified,
  title={Modified {S}chelling games},
  author={Kanellopoulos, P. and Kyropoulou, M. and Voudouris, A.},
  journal={Theoretical Computer Science},
  volume={880},
  pages={1--19},
  year={2021},
  publisher={Elsevier}
}

@article{mieszkowski1993causes,
  title={The causes of metropolitan suburbanization},
  author={Mieszkowski, P. and Mills, E.},
  journal={Journal of Economic Perspectives},
  volume={7},
  number={3},
  pages={135--147},
  year={1993},
  publisher={American Economic Association}
}

@article{ramani2021donut,
  title={The donut effect: How {COVID}-19 shapes real estate},
  author={Ramani, A. and Bloom, N.},
  journal={SIEPR Policy Brief, January},
  pages={1--8},
  year={2021}
}

@article{PaccagnanG24,
  author       = {Paccagnan, D. and Gairing, M.},
  title        = {In Congestion Games, Taxes Achieve Optimal Approximation},
  journal      = {Operations Research},
  volume       = {72},
  number       = {3},
  pages        = {966--982},
  year         = {2024},
  url          = {https://doi.org/10.1287/opre.2021.0526},
  doi          = {10.1287/OPRE.2021.0526},
  bibsource    = {dblp computer science bibliography, https://dblp.org}
}

@article{PV07,
  title={Schelling's spatial proximity model of segregation revisited},
  author={Pancs, R. and Vriend, N.},
  journal={Journal of Public Economics},
  volume={91},
  pages={1--24},
  year={2007},
}

@article{Plaza06,
  title={The return on investment of the {G}uggenheim {M}useum {B}ilbao},
  author={Plaza, B.},
  journal={International Journal of Urban and Regional Research},
  volume={30},
  number={2},
  pages={452--467},
  year={2006},
}

@article{Plaza07,
  title={The {B}ilbao effect ({G}ughenheim {M}useum {B}ilbao)},
  author={Plaza, B.},
  journal={Museum News},
  volume={86},
  number={5},
  year={2007},
}

@article{rosenthal1973class,
  title={A class of games possessing pure-strategy {N}ash equilibria},
  author={Rosenthal, R.},
  journal={International Journal of Game Theory},
  volume={2},
  pages={65--67},
  year={1973},
  publisher={Physica-Verlag}
}

@book{roughgarden2005selfish,
  title={Selfish Routing and the Price of Anarchy},
  author={Roughgarden, T.},
  year={2005},
  publisher={MIT Press}
}

@article{roughgarden2002bad,
  title={How bad is selfish routing?},
  author={Roughgarden, T. and Tardos, {\'E}.},
  journal={Journal of the ACM},
  volume={49},
  number={2},
  pages={236--259},
  year={2002},
  publisher={ACM New York, NY, USA}
}

@article{roughgarden2015intrinsic,
  author    = {Roughgarden, T},
  title     = {Intrinsic robustness of the price of anarchy},
  journal   = {Journal of the {ACM}},
  volume    = {62},
  number    = {5},
  pages     = {32:1--32:42},
  year      = {2015},
}

@article{SZ83,
  title={The Prevalence of {B}raess' Paradox},
  author={Steinberg, R. and Zangwill, W.},
  journal={Transportation Science},
  volume={17},
  number={3},
  pages={301--318},
  year={1983},
}

@article{su2020intriguing,
  title={Intriguing effects of underlying star topology in {S}chelling's model with blocks},
  author={Su, G. and Xiong, Q. and Zhang, Y.},
  journal={Physical Review E},
  volume={102},
  number={1},
  pages={012317},
  year={2020},
  publisher={APS}
}

@article{su2023significant,
  title={Significant suppression of segregation in {S}chelling’s metapopulation model with star-type underlying topology},
  author={Su, G. and Zhang, Y.},
  journal={The European Physical Journal B},
  volume={96},
  number={7},
  pages={91},
  year={2023},
  publisher={Springer}
}

@article{schelling1971dynamic,
  title={Dynamic models of segregation},
  author={Schelling, T.},
  journal={Journal of Mathematical Sociology},
  volume={1},
  number={2},
  pages={143--186},
  year={1971},
  publisher={Taylor \& Francis}
}

@article{Zhang04,
  title={Residential segregation in an all-integrationist world},
  author={Zhang, J.},
  journal={Journal of Economic Behavior and Organization},
  volume={54},
  pages={533--555},
  year={2004},
}

@article{Zhang11,
  title={Tipping and residential segregation: A unified {S}chelling model},
  author={Zhang, J.},
  journal={Journal of Regional Science},
  volume={51},
  pages={167--193},
  year={2011},
}

\section*{Appendix}

Here in the appendix we give the formal details underlying the
example of the donut effect described in Section~\ref{sec:case-study}.
The example consists of 1 red large inner-city neighbourhood with capacity 1000
and 10 blue small suburban neighbourhoods each with capacity 100. 
(For motivation, the reader may assume the units are measured in thousands
so the inner city has a maximum population capacity of one million,
and the ten suburbs have a combined maximum total population capacity of one million.)
The utility functions for the red neighbourhood and for the blue neighbourhoods are:

\noindent
\begin{minipage}{.43\linewidth}
$$
f_R(x)
=\begin{cases}
			25 + \frac{3}{20} x & \text{if } 0 \le x \le 500\\
            200-\frac{1}{5}x & \text{if } 500 \le x \le 1000
		 \end{cases}
$$
\end{minipage}
\quad
\begin{minipage}{.45\linewidth}
$$
f_B(x)
=\begin{cases}
			26+ \frac{3}{20} x & \text{if } 0 \le x \le 50\\
            	101 - \frac{43}{31} (x-50) & \text{if } 50 \le x \le 81\\
            301-3x & \text{if } 81 \le x \le 100
		 \end{cases}
$$
\end{minipage}

These utility functions are illustrated in Figure~\ref{fig:urban-flight}.
Observe that the neighbourhoods are
similar except the blue neighbourhoods have a kinked downward slope.
Note that the kink is at the point $(81,58)$.
\begin{figure}[h!]
\centering
\newcommand{\dd}{0.01}
\begin{tikzpicture}[scale=0.77,xscale=5,yscale=4,domain=0:1]
  \draw[->] (-0.02,0) -- (1.05,0) node[right] {$x$};
  \draw[->] (0,-0.02) -- (0,1.08+2*\dd) node[above] {$f_R(x)$};

  \draw (0,0.02) -- (0,-0.02) node[below] {$0$};
  \draw (0.5,0.02) -- (0.5,-0.02) node[below] {$500$};
  \draw (1,0.02) -- (1,-0.02) node[below] {$1000$};
  
  \draw[color=red, ultra thick]   (0,0.25) -- (0.5,1);
  \draw[color=red, ultra thick]   (0.5,1) -- (1,0);
    \draw (0.02,0.25) -- (-0.02, 0.25) node[above left = -2mm and 0mm]{$25$};
  \draw (0.02,1) -- (-0.02, 1) node[left]{100};

  \draw[thin, dashed, color=gray] (0.5,-0.02) -- (0.5,1.02+\dd);
  \draw[thin, dashed, color=gray] (1,-0.02) -- (1,1);
  \draw[thin, dashed, color=gray] (-0.02,1) -- (1.02,1);
\end{tikzpicture}
\quad
\begin{tikzpicture}[scale=0.77,xscale=5,yscale=4,domain=0:1]
  \draw[->] (-0.02,0) -- (1.05,0) node[right] {$x$};
  \draw[->] (0,-0.02) -- (0,1.08+2*\dd) node[above] {$f_B(x)$};

  \draw (0,0.02) -- (0,-0.02) node[below] {$0$};
  \draw (0.5,0.02) -- (0.5,-0.02) node[below] {$50$};
   \draw (0.81,0.02) -- (0.81,-0.02) node[below] {$81$};
  \draw (1,0.02) -- (1,-0.02) node[below] {$100$};

  \draw (0.02,\dd) -- (-0.02, \dd) node[left]{$1$};
    \draw (0.02,0.25+\dd) -- (-0.02, 0.25+\dd) node[left]{$26$};
  \draw (0.02,1+\dd) -- (-0.02, 1+\dd) node[left]{$101$};

  \draw[color=blue,ultra thick]   (0,0.25+\dd) -- (0.5,1+\dd);
  \draw[color=blue, ultra thick]   (0.5,1+\dd) -- (0.81, 0.57+\dd) --(1,\dd);

  \draw[thin, dashed, color=gray] (0.5,-0.02) -- (0.5,1.02+\dd);
  \draw[thin, dashed, color=gray] (1,-0.02) -- (1,1.02+\dd);
  \draw[thin, dashed, color=gray] (-0.02,1+\dd) -- (1.02,1+\dd); 
    \draw[thin, dashed, color=gray] (0.81,0) -- (0.81,0.57+\dd); 
\end{tikzpicture}
\caption{An instance with urban flight.}\label{fig:urban-flight}
\end{figure}
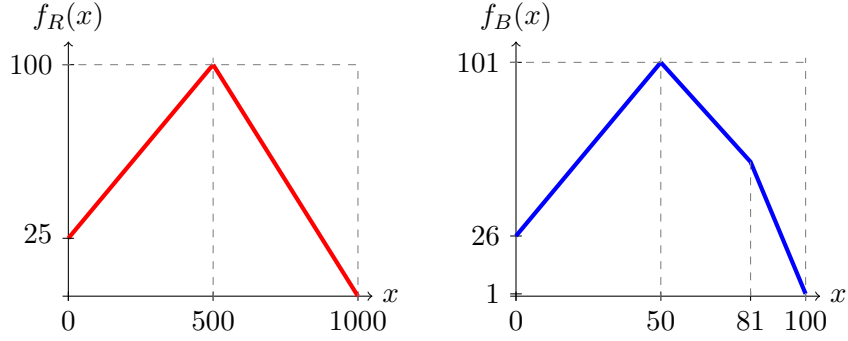

Assume the total population of the city is $1000$ units (i.e. one million people), and let's assume that the inner-city and suburban neighbourhoods
all begin at $50\%$ capacity.
This example is troubling because a dynamic arises where agents in the 
inner-city continue to leave for the suburbs over time until the equilibrium is reached. In this unique equilibrium the inner-city is sparsely populated (at $16\%$ capacity) and the suburbs are densely populated (at $84\%$ capacity). This is despite the suburbs and inner-city being very similar in quality, as shown by their similar utility functions.

To prove this, assume that $10x$ units of agents are in the red neighbourhood and $100-x$ units of agents are in each blue 
neighbourhood; we remark that at any equilibrium the quantity in each blue neighbourhood will be the same.
Now, for $50\ge x\ge 19$, we have
\begin{eqnarray*}
f_B(100-x)-f_R(10x) &=& \left( 101 - \frac{43}{31} \cdot((100-x)-50)\right) - \left(25+ \frac{3}{20}\cdot 10x\right) \\
&=& 76 +\frac{43}{31}\cdot(50-x)-\frac32 x \\
&=& 145.355-2.887x\\
&\ge& 145.355-2.887\cdot 50\\
&>& 0
\end{eqnarray*}
Thus agents obtain a higher utility by moving from the inner-city red neighbourhood to the suburban blue neighbourhoods as long as $x\ge 19$.
In fact, this process continues until $x=16$. Specifically,
for $19\ge x> 16$, we have
\begin{eqnarray*}
f_B(100-x)-f_R(10x) &=& \left( 301 - 3\cdot(100-x)\right) - \left(25+ \frac{3}{20}\cdot 10x\right) \\
&=& \frac32 x -24\\
&>& \frac32 \cdot 16 -24\\
&=& 0
\end{eqnarray*}
Observe this process terminates at $x=16$ where we obtain the unique Nash equilibrium with $x_R=10\cdot 16=160$ units of population in the inner-city and $x_B=100-16=84$ units of population in each suburban neighbourhood.

This equilibrium indeed gives the donut picture shown
in~Figure~\ref{fig:city_initial}. How can we rectify this donut effect using a small targeted investment?
To do this consider what happens if the red inner-city neighbourhood is improved slightly to:
$$
g_R(x)
=\begin{cases}
			32 + \frac{69}{500} x & \text{if } 0 \le x \le 500\\
            202-\frac{101}{500}x & \text{if } 500 \le x \le 1000
		 \end{cases}
$$
This new utility function is show by the dashed line in Figure~\ref{fig:no-urban-flight}. The shaded grey area denotes the corresponding investment cost
for this improvement.
\begin{figure}[h!]
\centering
\newcommand{\dd}{0.01}
\begin{tikzpicture}[scale=0.77,xscale=5,yscale=4,domain=0:1]
  \draw[->] (-0.02,0) -- (1.05,0) node[right] {$x$};
  \draw[->] (0,-0.02) -- (0,1.08+2*\dd) node[above] {$g_R(x)$};

  \draw (0,0.02) -- (0,-0.02) node[below] {$0$};
  \draw (0.5,0.02) -- (0.5,-0.02) node[below] {$500$};
  \draw (1,0.02) -- (1,-0.02) node[below] {$1000$};
  
  \draw[color=red, thick]   (0,0.25) -- (0.5,1);
  \draw[color=red, thick]   (0.5,1) -- (1,0);
  \draw (0.02,0.25+7*\dd) -- (-0.02, 0.25+7*\dd) node[above left = -2mm and 0mm]{$32$};
  \draw (0.02,1+\dd) -- (-0.02, 1+\dd) node[above left = -1mm and 0mm]{$101$};

  \draw[color=red,ultra thick,dashed]   (0,0.25+7*\dd) -- (0.5,1+\dd);
  \draw[color=red,ultra thick,dashed]   (0.5,1+\dd) -- (1,0);

  \draw[thin, dashed, color=gray] (0.5,-0.02) -- (0.5,1.02+\dd);
  \draw[thin, dashed, color=gray] (1,-0.02) -- (1,1.02+\dd);
  \draw[thin, dashed, color=gray] (-0.02,1) -- (1.02,1);
  \draw[thin, dashed, color=gray] (-0.02,1+\dd) -- (1.02,1+\dd); 

    \draw [fill=gray, opacity=0.4]
       (0, 0.25+7*\dd) -- (0.5,1+\dd)  -- (1,0) -- (0.5, 1) -- (0,0.25) -- cycle;
\end{tikzpicture}
\quad
\begin{tikzpicture}[scale=0.77,xscale=5,yscale=4,domain=0:1]
  \draw[->] (-0.02,0) -- (1.05,0) node[right] {$x$};
  \draw[->] (0,-0.02) -- (0,1.08+2*\dd) node[above] {$f_B(x)$};

  \draw (0,0.02) -- (0,-0.02) node[below] {$0$};
  \draw (0.5,0.02) -- (0.5,-0.02) node[below] {$50$};
   \draw (0.81,0.02) -- (0.81,-0.02) node[below] {$81$};
  \draw (1,0.02) -- (1,-0.02) node[below] {$100$};

  \draw (0.02,0.25+\dd) -- (-0.02, 0.25+\dd) node[left]{$26$};
  \draw (0.02,1+\dd) -- (-0.02, 1+\dd) node[left]{$101$};

  \draw[color=blue,ultra thick]   (0,0.25+\dd) -- (0.5,1+\dd);
  \draw[color=blue, ultra thick]   (0.5,1+\dd) -- (0.81, 0.57+\dd) --(1,\dd);

  \draw[thin, dashed, color=gray] (0.5,-0.02) -- (0.5,1.02+\dd);
  \draw[thin, dashed, color=gray] (1,-0.02) -- (1,1.02+\dd);
  \draw[thin, dashed, color=gray] (-0.02,1+\dd) -- (1.02,1+\dd); 
    \draw[thin, dashed, color=gray] (0.81,0) -- (0.81,0.57+\dd); 
\end{tikzpicture}
\caption{Improving the Inner-City.}\label{fig:no-urban-flight}
\end{figure}
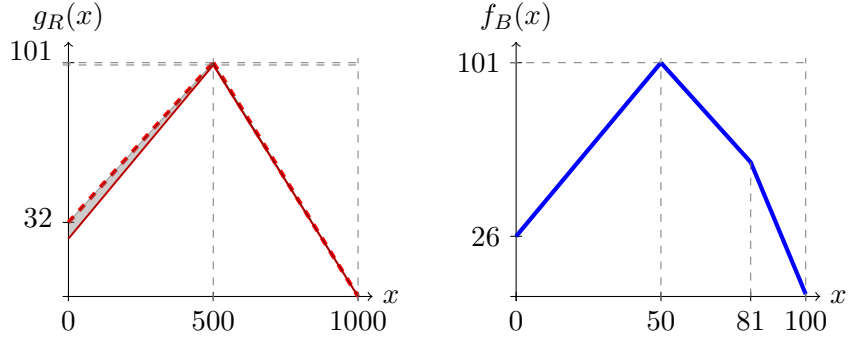

We claim the effect of this improvement is to create a reverse dynamic, with agents moving back from the suburbs to the inner-city. This dynamic
terminates when the inner-city neighbourhood and every suburban neighbourhood
are at exactly $50\%$ capacity.
To see this, consider first the case where $x\le 19$.
We then have
\begin{eqnarray*}
g_R(10x)-f_B(100-x) &=& \left( 32+\frac{69}{500}\cdot 10x\right) - \left(301-3\cdot (100-x)\right)\\
&=& 31-\frac{81}{50}x \\
&\ge& 31-\frac{81}{50}\cdot 19 \\
&>& 0
\end{eqnarray*}
It follows that while $x\le 19$, people will move back from the suburbs to the inner-city. What happens when $x\ge19$?
Then, for $19< x< 50$, we have
\begin{eqnarray*}
g_R(10x)-f_B(100-x) &=& \left( 32+\frac{69}{500}\cdot 10x\right) - \left( 101 - \frac{43}{31} \cdot((100-x)-50)\right)\\
&\ge& 0.354887 - 0.007097x \\
&\ge& 0.354887 - 0.007097\cdot 50\\
&>& 0
\end{eqnarray*}
This implies that agents still receive a higher utility by moving back from
the suburbs to the inner-city. This process continues until $x=50$ and we reach the unique Nash equilibrium where every neighbourhood is at exactly $50\%$ capacity. Thus the targeted investment caused a chain reaction 
that continues until an equal number of people live in the inner-city
as in total in the suburbs.
As claimed, the resultant equilibrium, shown in Figure~\ref{fig:city_improved}, has reversed the donut effect.

\end{document}